\documentclass[conference]{IEEEtran}

\usepackage[colorlinks=false,urlcolor=blue,citecolor=blue,linkcolor=blue,bookmarks=true, bookmarksopen=false,pdftitle=created by dvipdf,pdfcreator=NM, pdfauthor=Jain,pdfsubject=mor.sty:6.29.04]{hyperref}
\usepackage{amsmath,amsthm,amstext,amsfonts,amssymb,mathrsfs,euscript}
\usepackage{graphicx,subfigure,color,algorithm,algorithmic,fullpage,setspace}
\usepackage{dblfloatfix}
\usepackage{personal}
\usepackage{tikz}
\usetikzlibrary{arrows,shapes}

\IEEEoverridecommandlockouts
\allowdisplaybreaks

\title{
A Nash Equilibrium Need Not Exist in the Locational Marginal Pricing Mechanism
}

\author{
\IEEEauthorblockN{Wenyuan Tang}
\IEEEauthorblockA{
Department of Electrical Engineering\\
University of Southern California\\
{\tt wenyuan@usc.edu}}
\and
\IEEEauthorblockN{Rahul Jain}
\IEEEauthorblockA{
EE \& ISE Departments\\
University of Southern California\\
{\tt rahul.jain@usc.edu}}
}

\begin{document}

\maketitle
\thispagestyle{empty}
\pagestyle{empty}

\begin{abstract}
Locational marginal pricing (LMP) is a widely employed method for pricing electricity in the wholesale electricity market. Although it is well known that the LMP mechanism is vulnerable to market manipulation, there is little literature providing a systematic analysis of this phenomenon. In the first part of this paper, we investigate the economic dispatch outcomes of the LMP mechanism with strategic agents. We show via counterexamples, that contrary to popular belief,  a Nash equilibrium may not exist. And when it exists, the price of anarchy may be arbitrarily large. We then provide two sufficient conditions under either of which an efficient Nash equilibria exists. Last, we propose a new market mechanism for electricity markets, the Power Network Second Price (PNSP) mechanism that always induces an efficient Nash equilibrium. We briefly address the extensions on the demand side.
\end{abstract}


\section{Introduction}

Electric power is traded in wholesale electricity markets that involve various entities: the generators who generate and sell power, the distributors who buy power and sell to consumers, and the Independent System Operator (ISO). The ISO's role is to act as a dispatcher and a market-maker. The ISO asks for economic signals from the generators and the distributors and then determines a generation and dispatch schedule, as well as the nodal prices. The surplus is distributed among the Transmission System Operators (TSOs) to compensate them for use of the transmission infrastructure. This problem is commonly called the \emph{economic dispatch problem}. 

Indeed, the participants are economic agents each with their own objectives and some private information as well. Hence, they would be expected to behave in a way to further their own interests, or in other words, act strategically, and even give misleading information if it benefits them. However, in the design and analysis of mechanisms for economic dispatch, it is assumed that both generators and distributors truthfully reveal their marginal costs, marginal utilities, etc. Thus, an optimization problem is  formulated whose solution then gives the optimal dispatch schedule. 

In this paper, we start with the premise that the market participants, i.e., the generators and the distributors are self-interested and strategic economic agents, and despite the regulation they are subjected to, they can and will find ways to manipulate the market. This is not far-fetched since Enron energy traders indeed found ways to manipulate the congestion prices that led to the California electricity crisis of 2000--01 \cite{Mc02,Sw02}. Our goal is to understand what happens to market efficiency when there is strategic behavior, and the resulting outcomes. 


Locational marginal pricing (LMP), also known as nodal pricing, is a method for pricing electricity in the wholesale market \cite{KiSt04,ShYaLi02,Mo09,St02,Le03}. The locational marginal price (LMP) at a node is the cost of supplying the next increment of load at that node, taking into account transmission losses and congestion. Typically, it is the shadow price for the power balance constraint in the economic dispatch optimization problem. The mechanism is widely used in the US, Canada, Australia and New Zealand.

The LMP mechanism has been studied extensively but primarily in a competitive market setting where the market participants are assumed to be truthful about their cost and demand functions, and no one exercises market power \cite{Ne95}. In practice though,  exercise of market power is observed often, and there is no reason to believe that the revealed generation marginal costs are exactly truthful. To the best of our knowledge, ours seems to be the first work to study the LMP mechanism in a game-theoretic framework with strategic agents. 

Our work differs from the literature on supply function equilibria \cite{KlMe89,Gr99,BaGrKa04,JoTs11} in two ways. First, we take the underlying topology of the power network into account including any capacity limits on the transmission lines (which is precisely the reason LMPs are defined). Moreover, power flows over the network in accordance with the Kirchoff's circuit laws, which makes the game-theoretic anaysis of the network a lot more complicated than for communication or transportation networks. Secondly, we adopt a general bid format while quadratic forms are usually used in the study of supply function equilibrium.  


This paper is partly inspired by \cite{WuVaSpOr96} which formally established several ``folk theorems'' about the LMP mechanism, and showed that many common assertions (at that time) about it were factually incorrect. That paper also suggests the need to understand the LMP mechanism in a game-theoretic setting.  Among the few papers that study the strategic interaction in electricity markets involving transmission lines, \cite{Ho97} is based on a Cournot competition model, rather than the LMP mechanism. It also focuses on the strategic behavior of one agent only, assuming the environment at the other nodes are competitive. The motivation of \cite{HoMePa00} is close to our work but it focuses on the computational aspects of solving the optimization problem, and does not address the existence of Nash equilibria.

Our main results are the following: (i) Contrary to ``folk'' assertions, we present counterexamples to show that a Nash equilibrium may not even exist in the LMP mechanism; (ii) When a Nash equilibrium exists, the equilibrium outcome may be arbitrarily inefficient; (iii) We provide two sufficient conditions (including one on network topology) under either of which efficient Nash equilibria exist; (iv) We propose a new mechanism, called the Power Network Second Price (PNSP) mechanism that always induces an efficient Nash equilibrium; and (v) We consider extensions on the demand side in which the demand can be elastic and/or strategic.


\section{Economic Dispatch and the LMP mechanism}\label{sec:lmp}

We first introduce the economic dispatch problem. The concept of LMPs is based on the optimality conditions for this problem.

\subsection{The Economic Dispatch Problem}

We assume a connected power network throughout the paper, which consists of $I$ nodes (or buses), indexed by $i = 1,\ldots,I$. A transmission line connecting node $i$ and $j$ is characterized by its electrical admittance, denoted $Y_{ij} > 0$. If there is no such line, $Y_{ij} = 0$. Note that $Y_{ij} = Y_{ji}$.

Let $V_i$ be the magnitude of the voltage at node $i$, and $\th_i$ be the phase angle. The real power flow over the line from node $i$ to $j$ is given by
\begin{equation*}
q_{ij} = V_i V_j Y_{ij} \sin(\th_i-\th_j),
\end{equation*}
which ignores reactive power and line losses. By our sign convention, $q_{ij} = -q_{ji}$ is positive if the power flows from node $i$ to $j$. Also, it is reasonable to assume that $V_i$'s are approximately constant. Without loss of generality, we can set $V_i \equiv 1$. Furthermore, due to AC power flow, the economic dispatch problem is typically a nonlinear program that is difficult to solve in practice. Therefore, a DC flow model is often used as an approximation \cite{StJaAl09}, by assuming that the phase angle differences $|\th_i - \th_j|$ are small. Then we have $\sin(\th_i-\th_j) \approx (\th_i-\th_j)$ and
\begin{equation*}
q_{ij} = Y_{ij} (\th_i-\th_j).
\end{equation*}
Let the capacity limit of line $i$-$j$ be $C_{ij} = C_{ji} > 0$. So we have
\begin{equation*}
q_{ij} = Y_{ij} (\th_i-\th_j) \le C_{ij}.
\end{equation*}

Assume there are $N$ generators, indexed by $n = 1,\ldots,N$. Denote the set of generators at node $i$ by $\sN_i$. The cost of each generator $n$ is $c_n: \bbR_+ \to \bbR_+$ as a function of its generation $x_n$, where $c_n(0) = 0$, $c_n'(x_n) \ge 0$, and $c_n''(x_n) \ge 0$ (see Fig. \ref{fig1}). As for the demand side, we assume inelastic demand $D_i \ge 0$ at each node $i$. So the net power injected into the network at node $i$ is
\begin{equation*}
\sum_{n \in \sN_i} x_n - D_i = \sum_j q_{ij} = \sum_j Y_{ij} (\th_i-\th_j).
\end{equation*}

The economic dispatch problem is to determine the optimal generation schedule that minimizes the total cost subject to the transmission constraints. Formally, it is a convex program with linear constraints:
\begin{align}
& \underset{x_n,\th_i}{\min} & & \sum_n c_n(x_n) \label{eq:edp1}\\
& \text{s.t.} & & \sum_{n \in \sN_i} x_n - D_i = \sum_j Y_{ij} (\th_i-\th_j), \ \forall i, \label{eq:edp2}\\
& & & Y_{ij} (\th_i-\th_j) \le C_{ij}, \ \forall (i,j), \label{eq:edp3}\\
& & & x_n \ge 0, \ \forall n, \label{eq:edp4}
\end{align}
where (\ref{eq:edp2}) is the power balance equation, and (\ref{eq:edp3}) is the line flow constraint. For a more complete model of the economic dispatch problem (e.g., when line losses are taken into account), the reader can refer to \cite{Hs97}.

Note that the economic dispatch problem is based on a given set of states of the generators, in comparison with the unit commitment problem \cite{Mo09}; the $N$ generators in our consideration have been started up. The intercept of the cost function is indeed a sunk cost, which does not affect the optimal solution of the economic dispatch problem. One will also see that it is irrelevant for the optimal bidding strategies of the generators. This explains why we can assume $c_n(0) = 0$ (and later, $b_n(0) = 0$ for the bid curve) without loss of generality.

\subsection{The LMP Mechanism}

We now introduce how to determine the price at each node, specified by the LMP mechanism. Associate the Lagrange multipliers $\p_i$ with (\ref{eq:edp2}) and $\m_{ij} \ge 0$ with (\ref{eq:edp3}). The optimal solution of the economic dispatch problem (\ref{eq:edp1})-(\ref{eq:edp4}) is characterized by the following Karush-Kuhn-Tucker (KKT) conditions:
\begin{align*}
[c_n'(x_n)-\p_i] x_n & = 0, \ \forall n \in \sN_i,\\
c_n'(x_n)-\p_i & \ge 0, \ \forall n \in \sN_i,\\
\sum_j Y_{ij} (\p_i-\p_j+\m_{ij}-\m_{ji}) & = 0, \ \forall i,\\
\sum_{n \in N_i} x_n - D_i - \sum_j Y_{ij} (\th_i-\th_j) & = 0, \ \forall i,\\
\mu_{ij} [Y_{ij} (\th_i-\th_j) - C_{ij}] & = 0, \ \forall (i,j),\\
Y_{ij} (\th_i-\th_j) - C_{ij} & \le 0, \ \forall (i,j),\\
x_n & \ge 0, \ \forall n.
\end{align*}
The LMP at node $i$ is defined as $\p_i$, the interpretation of which is the cost of supplying the next increment of load at that node. Intuitively, the higher the LMP, the more difficult to deliver power to that node. The payoff of generator $n \in \sN_i$ is given by
\begin{equation}\label{eq:payoff}
u_n = \p_i x_n - c_n(x_n).
\end{equation}
The LMP mechanism is considered economically efficient in the following sense: if the price at each node $i$ is fixed as $\p_i$, and each generator $n$ chooses $x_n$ to maximize his own payoff, then the resulting dispatch is the economic dispatch. This can be seen from the KKT conditions, which state that the marginal cost of a generator with positive generation is exactly the LMP at that node.

There are some counterintuitive facts about LMP, due to Kirchhoff's circuit laws (under DC approximation). But this is outside the scope of this paper, and the reader can refer to \cite{WuVaSpOr96} for details.

\begin{figure}
\centering
\begin{tikzpicture}
\draw [->] (0,0) node[below left] {$O$} -- (6,0) node[right] {$x_n$};
\draw [->] (0,0) -- (0,4.5);
\draw (0,0) parabola (5,2) node[right] {$c_n(x_n)$};
\draw [-] (0,0) -- (1.2,0.3) node[draw,circle,inner sep=1pt,fill] {} -- (2.4,0.9) node[draw,circle,inner sep=1pt,fill] {} -- (3.6,1.8) node[draw,circle,inner sep=1pt,fill] {} -- (4.8,3) node[draw,circle,inner sep=1pt,fill] {} node[right] {$b_n^0(x_n)$};
\draw [-] (0,0) -- (2.8,2) node[draw,circle,inner sep=1pt,fill] {} -- (4.2,3.8) node[right] {$b_n(x_n)$};
\draw (1.4,1.4) node {$p_n$};
\draw (3.4,3.2) node {$q_n$};
\draw [dashed] (2.8,2) -- (2.8,0) node[below] {$s_n$};
\end{tikzpicture}
\caption{The true cost curve $c_n(x_n)$, the bid curve $b_n^0(x_n)$ in practice, and the simplified bid curve $b_n(x_n)$ determined by the three-dimensional bid $(p_n,s_n,q_n)$.}
\label{fig1}
\end{figure}
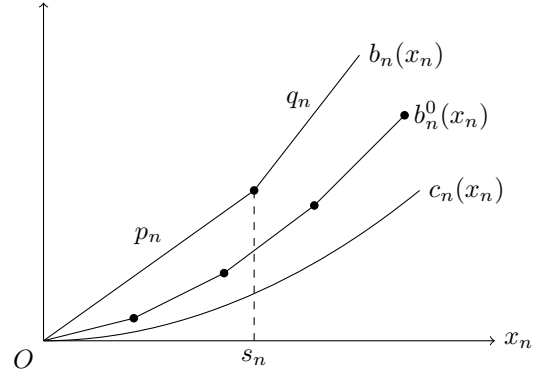

\section{Game-Theoretic Analysis}\label{sec:game}

The LMP mechanism is based on the assumption of a competitive environment, in which generators are considered as price takers. In reality, however, generators may have market power so that they may have an incentive to not reveal their cost functions truthfully. In that case, the economic dispatch problem is skewed and the solution may not be optimal.

The main focus of this paper is to study the equilibrium outcomes under LMP when generators act strategically. The first step is to reformulate the economic dispatch problem as a game.

\subsection{The Economic Dispatch Game}

Recall that the cost function space is infinite-dimensional. In the day-ahead market, each generator is only allowed to report a function from a finite-dimensional space as an approximation of his true cost function. Such a reported cost function is called a bid. We will specify the bid format in the next subsection.

We reformulate the economic dispatch problem as a game, which we call the economic dispatch game (or LMP game). Given a bid profile $b = (b_1,\ldots,b_N)$, the LMP mechanism determines the dispatch $x = (x_1,\ldots,x_N)$ as a solution of the following optimization problem:
\begin{align}
& \underset{x_n,\theta_i}{\text{min}} & & \sum_n b_n(x_n) \label{eq:edg1}\\
& \text{s.t.} & & \sum_{n \in \sN_i} x_n - D_i = \sum_j Y_{ij} (\th_i-\th_j), \ \forall i, \label{eq:edg2}\\
& & & Y_{ij} (\th_i-\th_j) \le C_{ij}, \ \forall (i,j), \label{eq:edg3}\\
& & & x_n \ge 0, \ \forall n. \label{eq:edg4}
\end{align}
Compared with the economic dispatch problem (\ref{eq:edp1})-(\ref{eq:edp4}), the only difference is that the true cost function $c = (c_1,\ldots,c_N)$ is replaced by the bid $b$. The LMP $\p = (\p_1,\ldots,\p_N)$ is defined similarly, depending on $b$ instead of $c$. The payoff is the same as (\ref{eq:payoff}). This completes the specification of the economic dispatch game.

We adopt the (pure) Nash equilibrium as the solution concept. As usual, a \emph{Nash equilibrium} is a bid profile $b$ with the associated outcome $(x,\p)$, in which no generator can be better off by a unilateral deviation. If $x$ also solves the economic dispatch problem (\ref{eq:edp1})-(\ref{eq:edp4}), then $b$ is called an \emph{efficient} Nash equilibrium.

Note that we stress the association between the bid $b$ and the outcome $(x,\p)$, because the primal optimal solution $x$ as well as the dual optimal solution $\p$ may not be unique for a given $b$, especially for flow constrained network and non-differentiable bids. This rarely happens in practice, but does cause technical issues in equilibrium analysis. For our purposes, we break ties for $\p$ by minimizing $\p_i$ in the lexicographic order. That is, from the set of the dual optimal solutions, we first choose those with the minimum $\p_1$, from which we then choose those with the minimum $\p_2$, and so forth.

\subsection{Bid Format}

In practice, the bid $b_n^0$ of generator $n$ is typically a piecewise linear function with increasing slopes, as shown in Fig. \ref{fig1}. According to the CAISO \cite{caiso}, for example, ``There are 10 bid segments and 11 associated bid points. Each bid point has a generation (MW) and price (PR) value, which are paired together as MW and price coordinates.''

In this paper, we adopt a simplified bid format. Specifically, the bid $b_n$ of generator $n$ is determined by a three-dimensional signal $(p_n,s_n,q_n)$, where $0 \le p_n \le q_n$ and $s_n \ge 0$:
\begin{equation*}
b_n(x_n) = \left\{\begin{array}{ll} p_nx_n, & x_n \le s_n,\\ q_nx_n+(p_n-q_n)s_n, & x_n > s_n.\end{array}\right.
\end{equation*}
This is illustrated in Fig. \ref{fig1}. Due to the correspondence, such a signal will also be called a bid. Note that when $p_n = q_n$, $b_n$ becomes a linear function so that the value of $s_n$ does not matter.

The underlying rationale of the simplification is that the (sub)gradient of $b_n^0$ at the resulting $x_n$ contains the key information to determine the outcome by convexity. Given $b^0 = (b_1^0,\ldots,b_N^0)$ with the resulting dispatch $x$, if $b_n^0$ is replaced by $b_n$ (specified by $(p_n,s_n,q_n)$) where $s_n = x_n$ and $p_n$ (or $q_n$) is the left (or right) derivative of $b_n^0$ at $x_n$, then the resulting dispatch remains the same. Clearly, $p_n < q_n$ if $x_n$ is a breakpoint of $b_n^0$ and $p_n = q_n$ otherwise.

Given the explicit bid format, we can restate the economic dispatch game (\ref{eq:edg1})-(\ref{eq:edg4}) as a linear program (where $x_n := x_n^p+x_n^q$):
\begin{align*}
& \underset{x_n^p,x_n^q,\theta_i}{\text{min}} & & \sum_n (p_n x_n^p + q_n x_n^q)\\
& \text{s.t.} & & \sum_{n \in \sN_i} (x_n^p+x_n^q) - D_i = \sum_j Y_{ij}(\th_i-\th_j), \ \forall i,\\
& & & Y_{ij}(\th_i-\th_j) \le C_{ij}, \ \forall (i,j),\\
& & & x_n^p \le s_n, \ \forall n,\\
& & & x_n^p, x_n^q \ge 0, \ \forall n.
\end{align*}

Such a three-dimensional bid is quite versatile. It will also be used in the proposed Power Network Second Price mechanism (see Section \ref{sec:pnsp}). Moreover, the idea of dimensional reduction may be applied to a more general underlying bid $b_n^0$, not necessarily a convex piecewise linear function.

We should also note that quadratic bid curves (as in supply function equilibrium literature) provide smooth dispatch, revenue and profit curves that facilitate calculus-based analysis, while piecewise linear bid curves (adopted by most ISOs) do not produce continuously differentiable dispatch, revenue and profit curves, requiring different analysis techniques \cite{CaAl04}. This issue, though, has little impact on our main results.

\section{Strategic Behavior in the LMP mechanism}\label{sec:cons}

In this section, we demonstrate the undesirable outcomes under the LMP mechanism, due to the selfish and strategic manipulation by generators. We present counterexamples to show that a Nash equilibrium may not exist in the economic dispatch game. Even when a Nash equilibrium exists, the price of anarchy may be arbitrarily large.

\subsection*{Non-Existence of Nash equilibria}

Game-theoretic analysis for multidimensional action space can be demanding in general. In the following examples, we sometimes turn our attention from the bid to the outcome (i.e., the dispatch and the LMPs), which may facilitate the analysis.

\begin{example}\label{exmp:1}
Consider the network as shown in Fig. \ref{fig2}. The network has two nodes, with generator $1$ at node $1$ and generator $2$ at node $2$. The line capacity is $C$. The only demand is $D > C$ at node $2$. The cost functions are $c_n(x_n) = a_nx_n$ for all $n$ with $a_1 < a_2$. So the efficient economic dispatch is $x^* = (C,D-C)$.

\begin{figure}[h]
\centering
\begin{tikzpicture}[>=latex]
\draw [-,ultra thick] (0,0) -- (0,1.6) node[above] {1};
\draw [-,ultra thick] (2.5,0) -- (2.5,1.6) node[above] {2};
\draw [-] (0,0.8) -- (1.25,0.8) node[above] {$C$} -- (2.5,0.8);
\node at (-1,1.2) [circle,draw] {$x_1$} edge (0,1.2);
\node at (3.5,1.2) [circle,draw] {$x_2$} edge (2.5,1.2);
\draw [->] (2.5,0.4) -- (3.3,0.4) node[right] {$D > C$};
\end{tikzpicture}
\caption{A two-node network for which no Nash equilibrium exists.}
\label{fig2}
\end{figure}
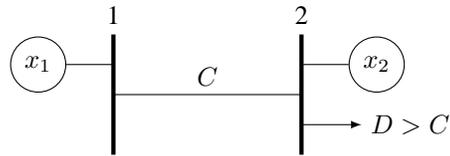

The LMP mechanism is used to do the dispatch and the generators now act strategically. But there is no Nash equilibrium in the corresponding game. This is because it is guaranteed that $x_2 \ge D-C$ and generator $2$ is able to make $\p_2$ arbitrarily large by choosing $p_2$ arbitrarily large, so that his payoff is unbounded. To make it non-trivial, we impose a bid constraint, say $q_1,q_2 \le a$ (which is a sufficiently large constant, interpreted as a reserve price of the demand). We now prove the non-existence by contradiction.

Suppose $b$ is a Nash equilibrium with the associated outcome $(x,\p)$. First, it can be seen that $a_1 \le \p_1 \le \p_2$ and $a_2 \le \p_2$. In fact, we have $a_2 \le \p_1 = \p_2$; otherwise, generator $1$ has an incentive to increase $\p_1$, so as to be better off. Moreover, $x_1 = C$; otherwise, generator $1$ has an incentive to deviate. Given $x_2 = D-C$, there must be $\p_2 = a$ and thus $\p_1 = a$. But this cannot be a Nash equilibrium, since generator $2$ now can increase $x_2$ by making $\p_2$ slightly smaller, so as to be better off. This proves that there does not exist a Nash equilibrium.
\end{example}

The above example might seem concocted since there is only one generator at the demand node, who knows that at least $D-C$ units of generation will have to be purchased from him. But as the next example shows, even if there is competition at the demand node, a Nash equilibrium still may not exist.

\begin{example}\label{exmp:2}
Consider the network as shown in Fig. \ref{fig3}. The network has three nodes, with generator $1$ at node $1$, generator $2$ at node $2$, and generator $3$ and $4$ at node $3$. The capacity of line $1$-$2$ is $C$, and the capacity of line $2$-$3$ is $C' > C$. The only demand is $D > C'$ at node $3$. The cost functions are $c_n(x_n) = a_nx_n$ for all $n$ with $a_1 < a_2 < a_3 < a_4$. So the efficient economic dispatch is $x^* = (C,C'-C,D-C',0)$.

\begin{figure}[h]
\centering
\begin{tikzpicture}[>=latex]
\draw [-,ultra thick] (0,0) -- (0,1.6) node[above] {1};
\draw [-,ultra thick] (2.5,0) -- (2.5,1.6) node[above] {2};
\draw [-,ultra thick] (5,-0.4) -- (5,2) node[above] {3};
\draw [-] (0,0.4) -- (1.25,0.4) node[below] {$C$} -- (3.75,0.4) node[below] {$C' > C$} -- (5,0.4);
\node at (1,1.2) [circle,draw] {$x_1$} edge (0,1.2);
\node at (3.5,1.2) [circle,draw] {$x_2$} edge (2.5,1.2);
\node at (6,1.7) [circle,draw] {$x_3$} edge (5,1.7);
\node at (6,0.8) [circle,draw] {$x_4$} edge (5,0.8);
\draw [->] (5,0) -- (5.8,0) node[right] {$D > C'$};
\end{tikzpicture}
\caption{A three-node network for which no Nash equilibrium exists.}
\label{fig3}
\end{figure}
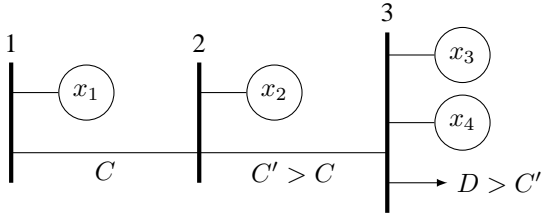

Suppose $b$ is a Nash equilibrium with the associated outcome $(x,\p)$. First, it can be shown that $a_3 \le \p_1 = \p_2 = \p_3$; otherwise, there must be some generator who has an incentive to deviate. Moreover, generator $1$ is able to ensure $x_1 = C$, and thus $x_2 = C'-C$. But this cannot be a Nash equilibrium, since generator $2$ now can increase $x_2$ by making $\p_2$ slightly smaller, so as to be better off. Therefore, there does not exist a Nash equilibrium.
\end{example}

\subsubsection*{Price of Anarchy}

There are also cases in which Nash equilibria exist but some of them are undesirable in terms of efficiency. The price of anarchy is a metric that measures how the efficiency degrades due to the selfish behavior of the players, compared with the socially optimal outcome. It is defined as the ratio between the cost of the worst equilibrium and the socially optimal cost:
\begin{equation*}
\text{PoA} = \frac{\max_x \sum_n c_n(x_n)}{\sum_n c_n(x_n^*)},
\end{equation*}
where $x^*$ is the economic dispatch that solves (\ref{eq:edp1})-(\ref{eq:edp4}), and $x$ is the resulting dispatch associated with any Nash equilibrium. The following example shows that the price of anarchy in the economic dispatch game can be arbitrarily large.

\begin{example}\label{exmp:3}
Consider the network as shown in Fig. \ref{fig4}. The network has two nodes, with generator $1$ and $4$ at node 1, and generator $2$ and $3$ at node $2$. The line capacity is $C$. The only demand is $D = 2C$ at node $1$. The cost functions are $c_1(x_1) = x_1$, $c_2(x_2) = kx_2$, $c_3(x_3) = kx_3$, and $c_4(x_4) = 2kx_4$, where $k > 1$ is a parameter subject to change. So the efficient economic dispatch is $x^* = (2C,0,0,0)$.

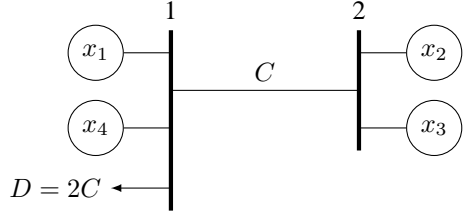
\begin{figure}[h]
\centering
\begin{tikzpicture}[>=latex]
\draw [-,ultra thick] (0,-0.8) -- (0,1.6) node[above] {1};
\draw [-,ultra thick] (2.5,0) -- (2.5,1.6) node[above] {2};
\draw [-] (0,0.8) -- (1.25,0.8) node[above] {$C$} -- (2.5,0.8);
\node at (-1,1.3) [circle,draw] {$x_1$} edge (0,1.3);
\node at (-1,0.3) [circle,draw] {$x_4$} edge (0,0.3);
\draw [->] (0,-0.5) -- (-0.8,-0.5) node[left] {$D = 2C$};
\node at (3.5,1.3) [circle,draw] {$x_2$} edge (2.5,1.3);
\node at (3.5,0.3) [circle,draw] {$x_3$} edge (2.5,0.3);
\end{tikzpicture}
\caption{A two-node network for which the price of anarchy can be arbitrarily large.}
\label{fig4}
\end{figure}

Consider the bid profile, $(p_n,s_n,q_n)$ for each $n$, where $p_1 = p_4 = 2k$, $p_2 = p_3 = k$ and $q_n = p_n$ (so that $s_n$ does not matter), with the associated outcome $x = (C,C,0,0)$ and $\p = (2k,k)$. It is easy to check that it is a Nash equilibrium. For example, generator $1$ has no incentive to make $\p_1 = k$ (so that $x_1 = 2C$) since his new payoff, $(k-1)2C$, would be smaller than his current payoff, $(2k-1)C$. Thus, the price of anarchy is bounded below by (since there may exist other equilibria that are even worse)
\begin{equation*}
\frac{c_1(C) + c_2(C)}{c_1(2C)} = \frac{C+kC}{2C} = \frac{k+1}{2} \to \infty,
\end{equation*}
as $k \to \infty$. Therefore, the price of anarchy can be arbitrarily large.
\end{example}

\section{Efficient NE: Sufficient Conditions}\label{sec:pros}

On the other hand, LMP works well in most cases. While the price of anarchy is unsatisfactory, the price of stability (defined as the ratio between the cost of the best equilibrium and the socially optimal cost) can still be good, which is equal to $1$ when an efficient Nash equilibrium (NE) exists. In this section, we present two sufficient conditions under either of which there exist efficient Nash equilibria.


Unless otherwise specified, we make the following assumption which ensures that no generator has market power to ask for  arbitrarily high prices.

\begin{assumption}\label{asmp:1}
The feasible dispatch set remains non-empty if any one of the generators is excluded.
\end{assumption}

Our first condition is the following. 

\begin{assumption}[Congestion-Free Condition]
No line flow constraint (\ref{eq:edp3}) is binding in the economic dispatch problem (\ref{eq:edp1})-(\ref{eq:edp4}).
\end{assumption}

The following lemma shows the uniformity of LMPs in the economic dispatch problem under this condition.

\begin{lemma}\label{lemma:1}
Under the congestion-free condition, all the LMPs are equal in the economic dispatch problem.
\end{lemma}

\begin{proof}
Let $\sI$ be the set of nodes with the largest LMPs, i.e., $\arg\max_i \p_i$. Under the congestion-free condition, $\m_{ij} = 0$ for all $(i,j)$. From the KKT conditions, we have
\begin{equation*}
\sum_j Y_{ij} (\p_i - \p_j) = 0, \ \forall i.
\end{equation*}
Then for each $i \in \sI$ and $j$ connected to $i$ (i.e., $Y_{ij} > 0$), since $\p_j \le \p_i$, we must have $\p_j = \p_i$, or $j \in \sI$. It follows by the connectedness of the network that all the nodes belong to $\sI$. Therefore, all the $\p_i$'s are equal.
\end{proof}

It is immediate to prove by construction the existence of efficient Nash equilibria in the economic dispatch game, under the congestion-free condition (which is defined for the economic dispatch problem).

\begin{theorem}
Under Assumption \ref{asmp:1} and the congestion-free condition, there exists an efficient Nash equilibrium in the LMP game.
\end{theorem}

\begin{proof}
Let $x^*$ be an economic dispatch that solves (\ref{eq:edp1})-(\ref{eq:edp4}), with the associated LMP $\p^*$, where $\p_i^* = \p_1^*$ for all $i$ (by Lemma \ref{lemma:1}). Consider the bid profile, $(p_n,s_n,q_n)$ for each $n$, where $p_n = q_n = \p_1^*$ for $n \in \sN_i$ (so that $s_n$ does not matter), with the associated outcome $x = x^*$ and $\p = \p^*$. It remains to show that it is a Nash equilibrium.

Consider generator $n$. From the KKT conditions, we have $x_n^* \in \arg\max_x \p_1^* x - c_n(x)$. His current payoff is $u_n = \p_1^* x_n^* - c_n(x_n^*)$. Suppose he changes his bid so that the resulting outcome is $(\tx,\tp)$. Under Assumption \ref{asmp:1} and given the others' bids, we have $\tp_i \le \p_1^*$. The payoff of generator $n$ will be
\begin{align*}
\tu_n & = \tp_i\tx_n - c_n(\tx_n)\\
& \le \p_1^* \tx_n - c_n(\tx_n)\\
& \le \p_1^* x_n^* - c_n(x_n^*)\\
& = u_n.
\end{align*}
Thus, he has no incentive to deviate. This proves that the constructed bid profile is a Nash equilibrium.
\end{proof}

Since there is no Nash equilibrium in Example \ref{exmp:1} and \ref{exmp:2}, the congestion-free condition cannot be satisfied in either of them (which is indeed true).


Our second condition is a condition on the network topology, and its satisfaction is easy to determine.

\begin{assumption}[Monopoly-Free Condition]
There are at least two generators at each node.
\end{assumption}

Note that under the monopoly-free condition, Assumption \ref{asmp:1} is automatically satisfied. The monopoly-free condition is more natural than the congestion-free condition, since we only need to know the placement of the generators in the network; we do not even need to know the line capacity limits, nor the cost functions of the generators. The proof of the existence of efficient Nash equilibria under the monopoly-free condition is similar as before.

\begin{theorem}
Under Assumption \ref{asmp:1} and  the monopoly-free condition, there exists an efficient Nash equilibrium in the LMP game.
\end{theorem}

\begin{proof}
Let $x^*$ be an economic dispatch that solves (\ref{eq:edp1})-(\ref{eq:edp4}), with the associated LMP $\p^*$. Consider the bid profile, $(p_n,s_n,q_n)$ for each $n$, where $p_n = q_n = \p_i^*$ for $n \in \sN_i$ (so that $s_n$ does not matter), with the associated outcome $x = x^*$ and $\p = \p^*$. It remains to show that it is a Nash equilibrium.

Consider generator $n \in \sN_i$. From the KKT conditions, we have $x_n^* \in \arg\max_x \p_i^* x - c_n(x)$. His current payoff is $u_n = \p_i^* x_n^* - c_n(x_n^*)$. Suppose he changes his bid so that the resulting outcome is $(\tx,\tp)$. Under the monopoly-free condition and given the others' bids, we have $\tp_i \le \p_i^*$. The payoff of generator $n$ will be
\begin{align*}
\tu_n & = \tp_i\tx_n - c_n(\tx_n)\\
& \le \p_i^* \tx_n - c_n(\tx_n)\\
& \le \p_i^* x_n^* - c_n(x_n^*)\\
& = u_n.
\end{align*}
Thus, he has no incentive to deviate. This proves that the constructed bid profile is a Nash equilibrium.
\end{proof}

In Example \ref{exmp:3}, $p_n = q_n = 1$ (and $s_n$ arbitrary) for all $n$, with the associated outcome $x = (2C,0,0,0)$ and $\p = (1,1)$, is such an efficient Nash equilibrium.

Although neither of the two conditions is necessary for a Nash equilibrium to exist (which can be easily shown), they are mild enough to cover most practical scenarios. In that sense, we justify the fact that the LMP mechanism is widely adopted and works well most of the time. In other words, we can say that the best equilibrium in the economic dispatch game is socially optimal except in rare cases.

\section{The Power Network Second Price Mechanism for Electricity Markets}\label{sec:pnsp}

Since the LMP mechanism does not always induce the desired outcome, we seek alternative mechanisms. In this section, we propose the Power Network Second Price (PNSP) market mechanism that always induces an efficient Nash equilibrium.

The proposed mechanism is a VCG-type mechanism. Note that the standard VCG mechanism does not apply directly, because we require a finite-dimensional (and preferably low-dimensional) bid space while the type space is infinite-dimensional.

In the PNSP mechanism, we still adopt the three-dimensional bid $(p_n,s_n,q_n)$ for generator $n$. The dispatch rule remains the same, given by the optimization problem (\ref{eq:edg1})-(\ref{eq:edg4}). The key difference is the payment rule (or pricing rule), as specified below.

Let $x^{-n_0} = (x_1^{-n_0},\ldots,x_N^{-n_0})$ denote the solution when generator $n_0$ is excluded (so that $x_{n_0}^{-n_0} = 0$). Note that we need Assumption \ref{asmp:1} to ensure that the definition is meaningful. The payment made to generator $n_0$ is given by
\begin{equation}\label{eq:payment}
w_{n_0} = \sum_{n \ne n_0} b_n(x_n^{-n_0}) - \sum_{n \ne n_0} b_n(x_n),
\end{equation}
which is the positive externality that generator $n_0$ imposes on the other players by his participation. Then, his payoff is
\begin{equation}\label{eq:payoff-pnsp}
u_{n_0} = w_{n_0} - c_{n_0}(x_{n_0}).
\end{equation}
This completes the specification of the PNSP mechanism.

We show that the PNSP mechanism always induces an efficient Nash equilibrium.

\begin{theorem}
Under Assumption \ref{asmp:1}, there exists an efficient Nash equilibrium in the PNSP mechanism specified by (\ref{eq:edg1})-(\ref{eq:payoff-pnsp}).
\end{theorem}

\begin{proof}
Let $x^*$ be an economic dispatch that solves (\ref{eq:edp1})-(\ref{eq:edp4}). Consider the bid profile, $(p_n,s_n,q_n)$ for each $n$, where $p_n = c_n'(x_n^*)$, $s_n = x_n^*$ and $q_n > p_n$. It is easy to check that $x^*$ also solves (\ref{eq:edg1})-(\ref{eq:edg4}) under this bid profile. It remains to show that it is a Nash equilibrium.

Consider generator $n_0$. His current payoff is
\begin{equation*}
u_{n_0} = \sum_{n \ne n_0} b_n(x_n^{-n_0}) - \sum_{n \ne n_0} b_n(x_n^*) - c_{n_0}(x_{n_0}^*).
\end{equation*}
Suppose he changes his bid, resulting in a new dispatch $\tx = (\tx_1,\ldots,\tx_N)$, with $\tx^{-n_0} = (\tx_1^{-n_0},\ldots,\tx_N^{-n_0})$ defined similarly as before. Then his payoff will be
\begin{equation*}
\tu_{n_0} = \sum_{n \ne n_0} b_n(\tx_n^{-n_0}) - \sum_{n \ne n_0} b_n(\tx_n) - c_{n_0}(\tx_{n_0}).
\end{equation*}
So his payoff changes by
\begin{align}
& \tu_{n_0} - u_{n_0} \NN\\
= \ & \sum_{n \ne n_0} b_n(x_n^*) - \sum_{n \ne n_0} b_n(\tx_n) + c_{n_0}(x_{n_0}^*) - c_{n_0}(\tx_{n_0}) \label{eq:pnsp1}\\
\le \ & \sum_{n \ne n_0} p_n(x_n^*-\tx_n) + c_{n_0}(x_{n_0}^*) - c_{n_0}(\tx_{n_0}) \label{eq:pnsp2}\\
= \ & \sum_{n \ne n_0} c_n'(x_n^*)(x_n^*-\tx_n) + c_{n_0}(x_{n_0}^*) - c_{n_0}(\tx_{n_0}) \NN\\
\le \ & c_{n_0}'(x_{n_0}^*)(\tx_{n_0}-x_{n_0}^*) + c_{n_0}(x_{n_0}^*) - c_{n_0}(\tx_{n_0}) \label{eq:pnsp3}\\
\le \ & 0. \label{eq:pnsp4}
\end{align}
Equation (\ref{eq:pnsp1}) follows since $x^{-n_0} = \tx^{-n_0}$. Equation (\ref{eq:pnsp2}) follows from the fact that $b_n(x_n^*) = p_nx_n^*$ and $b_n(\tx_n) \ge p_n\tx_n$. Since $x^*$ minimizes the convex objective function $c(x) = \sum_n c_n(x_n)$ over a convex set $P$ (determined by (\ref{eq:edp2})-(\ref{eq:edp4})), we have
\begin{equation*}
\nabla c(x^*) \cdot (x - x^*) \ge 0, \quad \forall x \in P.
\end{equation*}
In particular, letting $x = \tx$, we get
\begin{equation*}
\sum_n c_n'(x_n^*)(\tx_n-x_n^*) \ge 0,
\end{equation*}
from which (\ref{eq:pnsp3}) follows. Equation (\ref{eq:pnsp4}) follows from the property of convexity.

Thus, he has no incentive to deviate. This proves that the constructed bid profile is a Nash equilibrium.
\end{proof}

Like the LMP mechanism, there may also be undesirable Nash equilibria in the PNSP mechanism.

\section{Extensions on the Demand Side}\label{sec:ext}

In the modeling of the demand side, we need to consider two factors:
\begin{enumerate}
\item
Elastic or inelastic demand. The demand is inelastic if it is a constant; it is elastic when modeled as a valuation function.
\item
Strategic or non-strategic demand. The demand is non-strategic if the reserve price (for inelastic demand) or the valuation function (for elastic demand) is known; otherwise, it is strategic, in which case the consumers report (not necessarily truthfully) such private information as in a double-sided auction.
\end{enumerate}

Thus, there are four combinations of modeling the demand side. We have assumed so far that the demand at each node is inelastic and non-strategic. Due to space constraints, we illustrate the modeling for elastic and strategic demand.

Based on the single-sided economic dispatch problem (\ref{eq:edp1})-(\ref{eq:edp4}), we now assume there are $M$ consumers, indexed by $m = 1,\ldots,M$. Denote the set of consumers at node $i$ by $\sM_i$. The valuation of each consumer $m$ is $v_m: \bbR_+ \to \bbR_+$ as a function of its consumption $y_m$, where $v_m(0) = 0$, $v_m'(x_m) \ge 0$, and $v_m''(x_m) \le 0$. We obtain the double-sided economic dispatch problem in the following:
\begin{align*}
& \underset{x_n,y_m,\th_i}{\max} & & \sum_m v_m(y_m) - \sum_n c_n(x_n)\\
& \text{s.t.} & & \sum_{n \in \sN_i} x_n - \sum_{m \in \sM_i} y_m = \sum_j Y_{ij} (\th_i-\th_j), \ \forall i,\\
& & & Y_{ij} (\th_i - \th_j) \le C_{ij}, \ \forall (i,j),\\
& & & x_n, y_m \ge 0, \ \forall n, \forall m.
\end{align*}
The LMP and the payoff are defined similarly as before.

We can also define the double-sided economic dispatch game and ask each consumer to report a three-dimensional bid as an approximation of his true valuation function. We note that all the main results in this paper have their counterparts for the double-sided setting. The reader can refer to \cite{TaJa13} for our previous work, in which elastic and strategic demand is considered (with a slightly different bid format).

\section{Conclusion}\label{sec:conclu}

We provided a  framework for studying strategic interactions in economic dispatch via the LMP mechanism. We showed that contrary to folklore, a Nash equilibrium may not exist in the LMP market mechanism. And even when a NE exists, the price of anarchy may be arbitrarily large. What these results mean in practice is a very important question worthy of further investigation since wholesale electricity markets where the LMP mechanism is used seemingly work well. But the spot markets apparently do not. Perhaps the lack of an equilibrium in the wholesale market introduces instabilities in the spot market that wouldn't be there otherwise---we can only speculate! 

We have also shown that under two sufficient conditions (no congestion in the economic dispatch problem, or when there are at least two players at each node), an efficient Nash equilibrium does exist. Our findings coincide with the policy proposed in \cite{Sw02}: ensure enough competition in wholesale markets.

We also proposed a new market mechanism that always induces an efficient Nash equilibrium. In the double-sided setting, the mechanism can incur a budget deficit. Thus, further work is needed to tackle these difficult issues.


\bibliographystyle{IEEEtran}
\bibliography{biblio-smartgrid}

\begin{thebibliography}{10}
\providecommand{\url}[1]{#1}
\csname url@samestyle\endcsname
\providecommand{\newblock}{\relax}
\providecommand{\bibinfo}[2]{#2}
\providecommand{\BIBentrySTDinterwordspacing}{\spaceskip=0pt\relax}
\providecommand{\BIBentryALTinterwordstretchfactor}{4}
\providecommand{\BIBentryALTinterwordspacing}{\spaceskip=\fontdimen2\font plus
\BIBentryALTinterwordstretchfactor\fontdimen3\font minus
  \fontdimen4\font\relax}
\providecommand{\BIBforeignlanguage}[2]{{%
\expandafter\ifx\csname l@#1\endcsname\relax
\typeout{** WARNING: IEEEtran.bst: No hyphenation pattern has been}%
\typeout{** loaded for the language `#1'. Using the pattern for}%
\typeout{** the default language instead.}%
\else
\language=\csname l@#1\endcsname
\fi
#2}}
\providecommand{\BIBdecl}{\relax}
\BIBdecl

\bibitem{Mc02}
\BIBentryALTinterwordspacing
{McCullough Research}. (2002, June) Congestion manipulation in {ISO
  California}. [Online]. Available: \url{http://www.mresearch.com/pdfs/19.pdf}
\BIBentrySTDinterwordspacing

\bibitem{Sw02}
J.~L. Sweeney, \emph{The California Electricity Crisis}.\hskip 1em plus 0.5em
  minus 0.4em\relax Hoover Institution Press, 2002.

\bibitem{KiSt04}
D.~S. Kirschen and G.~Strbac, \emph{Fundamentals of power system
  economics}.\hskip 1em plus 0.5em minus 0.4em\relax John Wiley \& Sons, 2004.

\bibitem{ShYaLi02}
M.~Shahidehpour, H.~Yamin, and Z.~Li, \emph{Market operations in electric power
  systems}.\hskip 1em plus 0.5em minus 0.4em\relax Institute of Electrical and
  Electronics Engineers, Wiley-Interscience, 2002.

\bibitem{Mo09}
J.~A. Momoh, \emph{Electric power system applications of optimization}.\hskip
  1em plus 0.5em minus 0.4em\relax CRC Press, 2009.

\bibitem{St02}
S.~Stoft, \emph{Power System Economics: Designing Markets for
  Electricity}.\hskip 1em plus 0.5em minus 0.4em\relax Wiley-IEEE Press, 2002.

\bibitem{Le03}
F.~L{\'e}v{\^e}que, \emph{Transport pricing of electricity networks}.\hskip 1em
  plus 0.5em minus 0.4em\relax Kluwer Academic Publishers, 2003.

\bibitem{Ne95}
D.~Newbery, ``Power markets and market power,'' \emph{The Energy Journal},
  vol.~16, no.~3, pp. 36--66, 1995.

\bibitem{KlMe89}
P.~Klemperer and M.~Meyer, ``Supply function equilibria in oligopoly under
  uncertainty,'' \emph{Econometrica}, vol.~57, no. number, pp. 1243--1277,
  1989.

\bibitem{Gr99}
R.~Green, ``The electricity contract market in {England} and {Wales},''
  \emph{The Journal of Industrial Economics}, vol.~47, no.~1, pp. 107--124,
  1999.

\bibitem{BaGrKa04}
R.~Baldick, R.~Grant, and E.~Kahn, ``Theory and application of linear supply
  function equilibrium in electricity markets,'' \emph{Journal of Regulatory
  Economics}, vol.~25, no.~2, p. 143, 2004.

\bibitem{JoTs11}
R.~Johari and J.~N. Tsitsiklis, ``Parameterized supply function bidding:
  Equilibrium and efficiency,'' \emph{Operations Research}, vol.~59, no.~5, pp.
  1079--1089, 2011.

\bibitem{WuVaSpOr96}
F.~Wu, P.~Varaiya, P.~Spiller, and S.~Oren, ``Folk theorems on transmission
  access: Proofs and counterexamples,'' \emph{Journal of Regulatory Economics},
  vol.~10, no.~1, pp. 5--23, 1996.

\bibitem{Ho97}
W.~W. Hogan, ``A market power model with strategic interaction in electricity
  networks,'' \emph{The Energy Journal}, vol.~18, no.~4, p. 107, 1997.

\bibitem{HoMePa00}
B.~Hobbs, C.~Metzler, and J.-S. Pang, ``Strategic gaming analysis for electric
  power systems: an {MPEC} approach,'' \emph{IEEE Transactions on Power
  Systems}, vol.~15, no.~2, pp. 638--645, 2000.

\bibitem{StJaAl09}
B.~Stott, J.~Jardim, and O.~Alsac, ``{DC} power flow revisited,'' \emph{IEEE
  Transactions on Power Systems}, vol.~24, no.~3, p. 1290, 2009.

\bibitem{Hs97}
M.~Hsu, ``An introduction to the pricing of electric power transmission,''
  \emph{Utilities Policy}, vol.~6, no.~3, pp. 257--270, 1997.

\bibitem{caiso}
\BIBentryALTinterwordspacing
{The California ISO}. Data definitions for {ISO} public bid data. [Online].
  Available: \url{http://content.caiso.com/marketops/OASIS
  /pubbid/data/pbddocumentation.pdf}
\BIBentrySTDinterwordspacing

\bibitem{CaAl04}
M.~Cain and F.~Alvarado, ``Implications of cost and bid format on electricity
  market studies: linear versus quadratic costs,'' \emph{Large Engineering
  Systems Conference on Power Engineering}, pp. 2--6, 2004.

\bibitem{TaJa13}
W.~Tang and R.~Jain, ``Game-theoretic analysis of the nodal pricing mechanism
  for electricity markets,'' \emph{to appear in Proceedings of the 52nd IEEE
  Conference on Decision and Control}, 2013.

\end{thebibliography}

\end{document}